\documentclass[pdftex,pra,aps,onecolumn,twoside,nofootinbib,showkeys, showpacs]{revtex4}
\pdfoutput=1
\usepackage[utf8]{inputenc}
\usepackage[dvips]{graphicx}
\usepackage[stable]{footmisc}
\usepackage{amsmath,amsthm,amssymb}
\usepackage{youngtab}
\usepackage{mathtools}
\usepackage[all]{xy}
\usepackage{dsfont}
\usepackage{multirow}
\usepackage{rotating}
\usepackage{appendix}
\usepackage{color}

\def\squareforqed{\hbox{\rlap{$\sqcap$}$\sqcup$}}
\def\qed{\ifmmode\squareforqed\else{\unskip\nobreak\hfil
\penalty50\hskip1em\null\nobreak\hfil\squareforqed
\parfillskip=0pt\finalhyphendemerits=0\endgraf}\fi}

\def\duzomniejsze{<\kern-.7mm<}
\def\duzowieksze{>\kern-.7mm>}

\def\textbf#1{{\bf #1}}
\def\beq{\begin{equation}}
\def\eeq{\end{equation}}
\def\be{\begin{equation}}
\def\ee{\end{equation}}
\def\bal{\begin{align}}
\def\eal{\end{align}}
\def\ben{\begin{eqnarray}}
\def\een{\end{eqnarray}}
\def\beqa{\begin{eqnarray}}
\def\eeqa{\end{eqnarray}}
\def\eea{\end{array}}
\def\bea{\begin{array}}

\newcommand{\bei}{\begin{itemize}}
\newcommand{\eei}{\end{itemize}}
\newcommand{\bee}{\begin{enumerate}}
\newcommand{\eee}{\end{enumerate}}

\def\>{\rangle}
\def\<{\langle}

\def\ot{\otimes}

\newtheorem{lemma}{Lemma}
\newtheorem{corollary}{Corollary}
\newtheorem{proposition}{Proposition}
\newtheorem{theorem}{Theorem}
\newtheorem{definition}{Definition}

\newtheorem{remark}{Remark}
\newtheorem{example}{Example}

\theoremstyle{plain}

\def\bed{\begin{definition}}
\def\eed{\end{definition}}
\def\bel{\begin{lemma}}
\def\eel{\end{lemma}}
\def\bet{\begin{theorem}}
\def\eet{\end{theorem}}

\begin{document}

\title{Seperability
	Properties for a Class of Block Matrices.}
\author{Marek Mozrzymas$^1$, Adam Rutkowski$^{2,3}$, Micha{\l } Studzi{\'n}%
ski$^{2,3}$}
\affiliation{$^1$ Institute for Theoretical Physics, University of Wroc{\l }aw, 50-204
Wroc{\l }aw, Poland\\
$^2$ Faculty of Mathematics, Physics and Informatics, University of Gda{\'n}%
sk, 80-952 Gda{\'n}sk, Poland\\
$^3$ National Quantum Information Centre in Gda\'{n}sk, 81-824 Sopot, Poland}
\date{\today }

\begin{abstract}
It is shown that, for the block
matrices belonging to $M(nd,\mathbb{C})$ with commuting and normal block entries of dimension $d$, the separability of such
a block matrices is equivalent to its semi-positive definity. The separability
decomposition of lenght equal to the dimension of the block matrix (which is
smaller then Carathéodory theorem implies) is given. The separability decomposition depends only on eigenvalues of block entries in the first part and on eigenvectors of the block entries in the second part of the tensor product. It is shown that
semi-positive definity of considered block matrices is equivalent to
semi-positive definity $d$ smaller matrices of dimension $n$.
%Moreover we present other connections (not only separability properties) between set of block matrices and quantum physics. Speaking more precisely we investigate nonlocal correlations exhibited in quantum discord. We present conditions for vanishing of quantum discord  for block matrices.
%At the end of the paper as the special class of block matrices we  discuss new properties of Toeplitz matrices, which play important role in many aspects of modern science.
\end{abstract}

\maketitle

%%%%%%%%%%%%%%%%%%%%%

%%%%%%%%%%%%%%%%%%%%%

%%%%%%%%%%%%%%%%%%%%%%%%

\section{Introduction}
%korelacje sa wazne, warto wiedziec jakie kor ma dany stan, my skupiamy sie na szczegolnej klasie -> okazuja sie byc Toeplitzami

The phenomenon of quantum correlations, such as quantum entanglement, quantum discord or quantum steering show in the one of the most amazing way difference between calssical world and the quantum one. Even now, after eighty years from famous paper~\cite{Ein} new ideas and concepts arise  giving new opportunities for researchers, we mention here only the milestones such as quantum teleprotation~\cite{Bennett2}, quantum dense coding~\cite{Bennett3} or pioneering work in quantum cryptography~\cite{Ekert1} (for more applications see ~\cite{ryszh}). Because of the multiplicity of possible applications of various types of quantum correlations,  is required to know whether given quantum state exhibits desired type of correlations, or complementarily for example whether  is separable, when we have to deal with entanglement. Thanks to this checking separability of quantum states plays the key role  in quantum information theory. Despite of this all arguments of importance till now we do not have any general method(s), which allows us to checking whether given state is separable or constructing separable decomposition. Another notable fact is that even when we are given with some separable state it is really hard to present their explicit separable decomposition. There are only few non-trivial examples where such decomposition can be done in some general regime (see for example~\cite{Kie}), so every new  result for wide class of states is in our opinion welcome.

Here we focus on special class of quantum states represented by block matrices.
Namely in the paper~\cite{Gur} it has been shown that if a block Toeplitz matrix  if it is positive
semidfinite then it is separable, so in fact for these block matrices
separability is equivalent to semi-positivity. In this paper we prove that
for any block matrix belonging to $M(nd,\mathbb{C})$ (so of dimension $nd)$ with commuting and normal block entries, the
separability of  such a matrix is equivalent to its semi-positivity. The
structure of considered block matrice implies that its semi-positivity of
considered block matrices is equivalent to semi-positivity $d$ smaller
matrices belonging to $M(n,\mathbb{C})$, where $d$ is the block dimension and $n$ is the number of blocks in a
row of block matrix. 

Before we go further, let us say here a few words more about notation used in this manuscript to which we will refer in next sections. In this section and in our further considerations  by $\mathcal{B}(\mathcal{\mathbb{C}}^d)$ or by  $\mathcal{B}(\mathcal{H})$ we denote the algebra of all bounded linear operators on $\mathbb{C}^d$ or on $\mathcal{H}$. Using this notation let us define the following set:
\be
\label{Qset}
\mathcal{S}(\mathcal{H})=\{\rho \in \mathcal{B}(\mathcal{H}) \ | \ \rho \geq 0\},
\ee
which is set of all  unnormalised states~\footnote{If we consider separability it is enough to deal with unnormalized states.} on space $\mathcal{H}$. Let us now suppose  that we are given with two finite dimensional Hilbert spaces $\mathcal{K}, \mathcal{H}$. Matrix in the bipartite composition system $\rho \in \mathcal{S}(\mathcal{H}\ot \mathcal{K})$ is said to be separable whenever it can be written as $\rho=\sum_i  \rho_i \ot \sigma_i$, where $\rho_i,\sigma_i$ are unnormalised states on $\mathcal{H}$ and $\mathcal{K}$ respectively. If above conditions are not fulfilled  we say that matrix $\rho$ is entangled. 
%In quantum physics we demand additionally normalization condition, which requires $\sum_i p_i=1$ to have proper probability interpretation. Here for simplicity we consider unnormalised operators.

%As we mentioned in the abstract for the block
%matrices with commuting and normal block entries , the separability of such
%a block matrices is equivalent to its semi-positivity. This means that such class does not posses any nonclassical correlations which follow from quantum entanglement. However there are  other nonclassical correlations measured by so called quantum discord introduced in~\cite{Zurek}. Nonzero quantum discord indicates the presence of correlations arising as a result of from noncommutativity of quantum operators. Formally quantum discord for subsytem $A$ (in similar way for $B$) can be defined as:
%\be
%D_A(\rho_{AB})=S\left( \rho_A\right) -S\left( \rho_{AB}\right)+\mathop{\operatorname{min}}\limits_{\{\Pi^A\}}S\left(\rho_{B|\{\Pi^A\}} \right), 
%\ee
%where $S(\cdot)$ denotes von Neumann entropy and minimum is taken over all possible projective measurements on $A$.
%From above formula it is easy to see that this measure is in general asymetrical in that sense that $D_A(\rho_{AB})$ can differ from $D_B(\rho_{AB})$. One of the goals of this paper is to investigate properties the class of block matrices with respect to quantum discord.

The paper is organised as follows: In the Section~\ref{SII} we define our problem in the details and present the main results for this manuscript. Namely after introductory part we present  Proposition~\ref{P1} and Lemma~\ref{L1}, which are crucial to formulate the Theorem~\ref{T1}.  In the mentioned theorem we present  necessary and sufficient conditions for separability in the language of eigenvalues and eigenvectors of proper components. 
At the end of the Section~\ref{SII}  we give three examples in which we present how our main theorem works in practice. It is worth to mention here about Example~\ref{E2} which has explicit connection with know class of matrices, knowing as circulant ones~\cite{Chrus}. Manuscript contains also Appendix~\ref{appA} where all the most important proofs from Section~\ref{SII} are presented.

\section{Separability properties of block matrices.}
\label{SII}

In this section we focus on separable decomposition properties of block matrices. The main result is contained in Theorem~\ref{T1} when conditions if and only if for separability are formulated for certain block matrices. In this same theorem authors also present explicit separable decomposition for mentioned class.
At the end of this section we give also three exemplary examples, which show how our result works in practice.
We consider the following block matrix 
\be
\label{blockM}
T=\left( 
\begin{array}{ccccc}
B_{11} & B_{12} & B_{13} & \cdots & B_{1n} \\ 
B_{21} & B_{22} & B_{23} & \cdots & B_{2n} \\ 
\vdots & \vdots & \vdots & \ddots & \vdots \\ 
B_{n-1 1} & B_{n-1 2} & B_{n-1 3} & \cdots & B_{n-1 n} \\ 
B_{n1} & B_{n2} & B_{n3} & \cdots & B_{nn}
\end{array}
\right), 
\ee
where $B_{ij}\in M(d, \mathbb{C})$ for $i,j=1,\ldots,n$. We assume that the matrices $B_{ij}$ in $T$ are normal and
commute. Such a matrices have the following decomposition property:

\begin{proposition}
	\label{P1}
	Consider a block matrix as in equation~\eqref{blockM} which is not necessearily positive semidefinite and 
	where $B_{ij}\in M(d,\mathbb{C})$ for $i,j=1,\ldots,n$ are normal and all of them form a commutative set of
	matrices. From these assumptions it follows that there exists for $B_{ij}\in
	M(d,\mathbb{C}),$ for $i,j=1,\ldots,n$  a common set of orthonormal eigenvectors $
	\{u_{k}\}_{k=1}^{d}$ such that 
	\be
	B_{ij}u_{k}=\beta _{k}^{ij}u_{k},\quad i,j=1,\ldots,n,\quad k=1,\ldots,d, 
	\ee
	where $\{\beta _{k}^{ij}\}$ are eigenvalues of the matrix $B_{ij}$. Then we
	have the following decomposition of the matrix $T$ 
	\be
	\label{1}
	T=\sum_{k=1}^{d}M(\beta _{k}^{ij})\otimes u_{k}u_{k}^{\dagger},
	\ee
	where 
	\be
	M(\beta _{k}^{ij})=\left( 
	\begin{array}{ccccc}
	\beta _{k}^{11} & \beta _{k}^{12} & \beta _{k}^{13} & \cdots & \beta _{k}^{1n}
	\\ 
	\beta _{k}^{21} & \beta _{k}^{22} & \beta _{k}^{21} & \cdots & \beta _{k}^{2n}
	\\ 
	\vdots & \vdots & \vdots & \ddots & \vdots \\ 
	\beta _{k}^{n-1 i} & \beta _{k}^{n-1 2} & \beta _{k}^{n-1 3} & \cdots & \beta
	_{k}^{n-1n} \\ 
	\beta _{k}^{n1} & \beta _{k}^{n2} & \beta _{k}^{n3} & \cdots & \beta _{k}^{nn}%
	\end{array}%
	\right) \in M(n,\mathbb{C}) 
	\ee
	is a matrix whose entries are eigenvalues of matrces $B_{ij}$ for $i,j=1,\ldots,n$
	coresponding the eigenvector $u_{k}$ for $k=1,\ldots,d$ and $u_{k}u_{k}^{\dagger}\in
	M(d,\mathbb{C})$.
\end{proposition}

\begin{remark}
	\bigskip The matrix decomposition as in equation~\eqref{1} has the property that it can be written as a
	tensor product of eigen-vectors of matrices $B_{ij}$ and the eigen-values
	of $B_{ij}$, which are in different positions in the tensor product.
\end{remark}

The matrix decomposition of the form~\eqref{1} in the  Proposition~\ref{P1} has the
following intersesting property:

\begin{lemma}
	\label{L1}
	Let $\{u_{i}\}_{i=1}^{d}$ be an orthonormal basis in the space $\mathbb{C}^{d}$ and $M_{i}\in M(n,\mathbb{C})$ for $i=1,\ldots,d$ then the matrix 
	\be
	T=\sum_{k=1}^{d}M_{k}\otimes u_{k}u_{k}^{\dagger}\in M(nd,\mathbb{C}) 
	\ee
	is semi-positive definite if and only if all matrices $M_{k}$ are
	semi-positive definite i.e. 
	\be
	\forall k=1,\ldots,d\quad M_{k}\geq 0.
	\ee
\end{lemma}

Directly from Lemma~\ref{L1} follows

\begin{corollary}
	\label{C1}
	Let $\{u_{i}\}_{i=1}^{d}$ be an orthonormal basis in the space $\mathbb{C}^{d}$ and $M_{i}\in M(n,\mathbb{C})$ for $i=1,\ldots,d$ then if the matrix 
	\be
	T=\sum_{k=1}^{d}M_{k}\otimes u_{k}u_{k}^{\dagger}\in M(nd,
	\mathbb{C})
	\ee
	is semi-positive definite, then it is separable and has the following
	separability decomposition
	\be
	T=\sum_{k=1}^{d}\sum_{j=0}^{n-1}\lambda _{j}^{k}v_{j}^{k}v_{j}^{k\dagger}\otimes
	u_{k}u_{k}^{\dagger},
	\ee
	where $\{\lambda _{j}^{k}\}_{j=0}^{n-1}$ and $\{v_{j}^{k}\}_{j=0}^{n-1}$ are
	eigenvalues and eigenvectors of the semi-positive definite (so hermitian) matrix $
	M_{k}\in M(n,\mathbb{C})$. So in fact, for matrices with decomposition of the form~\eqref{1} separability
	equivalent to semi-positivity.
\end{corollary}
%\textcolor{red}{Dla zesp dod implikuje hrem}\\

\begin{remark}
	\label{C1a}
 Let us notice that for every matrix $X\in M(d,\mathbb{C})$ positivity of $X$ implies its hermiticity.	
\end{remark}
The above results allow to fotmulate the main statement of our paper.

\begin{theorem}
	\label{T1}
	Suppose that a block matrix $T$ 
	\be
	T=\left( 
	\begin{array}{ccccc}
	B_{11} & B_{12} & B_{13} & \cdots & B_{1n} \\ 
	B_{21} & B_{22} & B_{23} & \cdots & B_{2n} \\ 
	\vdots & \vdots & \vdots & \ddots & \vdots \\ 
	B_{n-1 1} & B_{n-1 2} & B_{n-1 3} & \cdots & B_{n-1 n} \\ 
	B_{n1} & B_{n2} & B_{n3} & \cdots & B_{nn}
	\end{array}
	\right) \in M(nd,\mathbb{C}) 
    \ee
	is such that all matrices $B_{ij}\in M(d,\mathbb{C})$ are normal and commute, then:
	\begin{enumerate}
		\item the matrix $T$ is separable, if and only if it is semi-positive definite,
		\item the matrix $T$ is semi-positive definite if and only if all $d$
		matrices $M(\beta _{k}^{ij})\in M(n,\mathbb{C})$ in its decomposition 
		\be
		\label{max}
		T=\sum_{k=1}^{d}M(\beta _{k}^{ij})\otimes u_{k}u_{k}^{\dagger}
		\ee
		are semi-positive definite,
		\item if the matrix $T$ is separable, then it has the following separability
		decomposition 
		\be
		\label{ds}
		T=\sum_{k=1}^{d}\sum_{j=0}^{n-1}\lambda _{j}^{k}v_{j}^{k}v_{j}^{k\dagger}\otimes
		u_{k}u_{k}^{\dagger}, 
		\ee
		where $\{\lambda _{j}^{k}\}_{j=0}^{n-1}$ and $\{v_{j}^{k}\}_{j=0}^{n-1}$ are
		eigenvalues and eigenvectors of matrices $M(\beta _{k}^{ij})$ and 
		\be
		B_{ij}u_{k}=\beta _{k}^{ij}u_{k},\quad i,j=1,\ldots,n,\quad k=1,\ldots,d, 
		\ee
		i.e. $\beta _{k}^{ij}$ and $u_{k}$ are eigenvalues and eigenvectors of $
		B_{ij}$.
	\end{enumerate}
	So for this class of block matrices the separbility decomposition may by
	constructed directly from the eigenvectors of matrices $B_{ij}\in M(d,\mathbb{C})$ and from eigenvalues and eigenvectors of matrices $M(\beta _{k}^{ij})\in
	M(n,\mathbb{C})$, which are positive semi-definite if $T$ is separable. 
\end{theorem}
%\textcolor{red}{remark w stosunku do tw Cata}\\

\begin{remark}
	\label{R1a}
	Theorem~\ref{T1} implies that maximal length of separable decomposition given by formula~\eqref{max} is bounded by $dn$. Reader notices that mentioned maximal length is smaller than bound given by Carathéodory theorem, which is in this case equal to $d^2n^2$.
\end{remark}

\begin{remark}
	\label{R2}
	\bigskip If the block matrix $T$ has some special structure, for example it
	is a Toeplitz matrix, then the same structure appears in the eigen-value
	matrices $M(\beta _{k}^{ij})$, which will be seen in examples below.
\end{remark}
%\textcolor{red}{war kon na dod macierzy herm}\\
%\textcolor{red}{Prop 2+zdanie przed
	%of matricesM(b i jk) is the foll , kawnt %dla kazdego ij + po dwa indesksy w %szacowaniu ij<ii}
Directly from elementary properties of semi-positive matrices it follows

\begin{proposition}
	\label{PaR2}
	A necessary condition on the eigenvalues of the matrices $B_{ij}$ for $
	i,j=0,1,\ldots,n-1$ for positivity of matrices $M(\beta_k^{ij}) $ is the
	following
	\be
	\forall k=1,\ldots,d\quad \forall i,j=1,\ldots,n-1\quad |\beta _{k}^{ij}|\leq \beta
	_{k}^{ii}. 
	\ee
\end{proposition}

Let us consider some examples of of block matrices in the considered class.

\begin{example}
	\label{E1}
	Let $B\in M(d,\mathbb{C})$ be a normal matrix and $P_{ij}\in \mathbb{C}\lbrack x]$ for $i,j=1,\ldots,n$ be arbitrary polynomials, then the block matrix $T$
	\be
	T=\left( 
	\begin{array}{ccccc}
	P_{11}(B)& P_{12}(B) & P_{13}(B) & \cdots & P_{1n}(B) \\ 
	P_{21}(B) & P_{22}(B) & P_{23}(B) & \cdots & P_{2n}(B)\\ 
	\vdots & \vdots & \vdots & \ddots & \vdots \\ 
	P_{n-1 i}(B) & P_{n-1 2}(B) & P_{n-1 3}(B) & \cdots & 
	P_{n-1 n}(B) \\ 
	P_{n1}(B) & P_{n2}(B) & P_{n3}(B) & \cdots & P_{nn}(B)
	\end{array}
	\right) =(P_{ij}(B))
	\ee
	satisfies the assumpions of Theorem~\ref{T1}, such a matrix is separable if and only it is semi-positive. The eigen-value matrices $M(\beta _{k}^{ij})$ have
	the following form  
	\be
	M(\beta_{k}^{ij})=\left( 
	\begin{array}{ccccc}
	P_{11}(\beta_{k}) & P_{12}(\beta_{k}) & P_{13}(\beta_{k}) & \cdots & P_{1n}(\beta_{k}) \\ 
	P_{21}(\beta_{k}) & P_{22}(\beta_{k}) & P_{23}(\beta_{k}) & \cdots & P_{2n}(\beta _{k}) \\ 
	\vdots & \vdots & \vdots & \ddots & \vdots \\ 
	P_{n-1 i}(\beta_{k}) & P_{n-1 2}(\beta_{k}) & 
	P_{n-1 3}(\beta_{k}) & \cdots & P_{n-1 n}(\beta_{k}) \\ 
	P_{n1}(\beta_{k}) & P_{n2}(\beta _{k}) & P_{n3}(\beta _{k}) & \cdots & P_{nn}(\beta _{k})
	\end{array}
	\right),
	\ee
	where $\{\beta_{k}\}_{k=1}^{d}$ are eigenvalues of the matrix $B$.
	The matrices $M(\beta_{k}^{ij})$ have  the same structure as the structure
	of blocks in the matrix $T$. The matrix $T$ is semi-positive definite iff
	all matrices $M(\beta_{k}^{ij})$ are positive semidefinite and then the 
	matrix $T$  is separable with the following separability decomposition
	\be
	T=\sum_{k=1}^{d}\sum_{j=0}^{n-1}\lambda _{j}^{k}v_{j}^{k}v_{j}^{k\dagger}\otimes
	u_{k}u_{k}^{\dagger},
	\ee
	where 
	\be
	Bu_{k}=\beta _{k}u_{k},\quad k=1,\ldots,d,
	\ee
	i.e.  $u_{k}$ are eigenvalues and eigenvectors of $B$ and $\{\lambda
	_{j}^{k}\}_{j=0}^{n-1}$ and $\{v_{j}^{k}\}_{j=0}^{n-1}$ are eigenvalues and
	eigenvectors of the matrices $M(\beta _{k}^{ij})$
\end{example}

\begin{example}
	\label{E2a}
	Let us consider the particular case of the previous example. 
	\be
	T=\left( 
	\begin{array}{ccccc}
	\mathbf{1} & B & B^{2} & \cdots & B^{n-1} \\ 
	\left( B\right)^{\dagger} & \mathbf{1} & B & \ldots & B^{n-2} \\ 
	\vdots & \vdots & \vdots & \ddots & \vdots \\ 
	\left( B^{n-2}\right)^{\dagger} & \left( B^{n-3}\right)^{\dagger} & \left(
	B^{n-4}\right)^{\dagger} & \cdots & B \\ 
	\left( B^{n-1}\right)^{\dagger} & \left( B^{n-2}\right)^{\dagger} & \left(
	B^{n-3}\right)^{\dagger} & \cdots & \mathbf{1}
	\end{array}
	\right) ,
	\ee
	where the matrix $B$ is normal. It is clear that in this case all block
	matrix entries commute and we have 
	\be
	M(\beta_k)=\left( 
	\begin{array}{ccccc}
	1 & \beta _{k} & (\beta _{k})^{2} & \cdots & (\beta _{k})^{n-1} \\ 
	\overline{\beta }_{k} & 1 & \beta _{k} & \cdots & (\beta _{k})^{n-2} \\ 
	\vdots & \vdots & \vdots & \ddots & \vdots \\ 
	(\overline{\beta }_{k})^{n-2} & (\overline{\beta }_{k})^{n-3} & (\overline{%
		\beta }_{k})^{n-4} & \cdots & \beta _{k} \\ 
	(\overline{\beta }_{k})^{n-1} & (\overline{\beta }_{k})^{n-2} & (\overline{%
		\beta }_{k})^{n-3} & \cdots & 1
	\end{array}
	\right).
	\ee
	One can check that 
	\be
	\forall n\in \mathbb{N}\quad \det (M(\beta_k))=(1-|\beta_{k}|^{2})^{n-1}.
	\ee
	Therefore from the Sylvester's criterion we it follows that if 
	\be
	\forall k=1,\ldots,d\quad |\beta _{k}|^{2}\leq 1,
	\ee
	then all matrices $M(\beta_k)$ are positive semidefinite and the matrix $T$
	in this example is separable with the following separability decomposition
	\be
	T=\sum_{k=1}^{d}M(\beta_k)\otimes u_{k}u_{k}^{\dagger},
	\ee
	where 
	\be
	Bu_{k}=\beta _{k}u_{k},\quad k=1\ldots,d,
	\ee
	i.e. $\beta_{k}$ and $u_{k}$ are eigenvalues and eigenvectors of $B$. Note
	that the unitary matrices (i.e. when $B\in U(d)$) satisfies the condition $
	\forall k=1,\ldots,d\quad |\beta_{k}|^{2}\leq 1$ and for unitary matrices all
	matrices $M(\beta_k)$ are of rank one.
\end{example}

The next example is more explicite.

\begin{example}
	\label{E2}
	In the paper~\cite{Chrus} a class of circulat matrices has been
	introduced, which are of the form 
	\be
	T=\sum_{k=0}^{d-1}\sum_{i,j=0}^{d-1}a_{ij}^{k}e_{ij}\otimes
	e_{i+k \ j+k}=\sum_{k=0}^{d-1}\sum_{i,j=0}^{d-1}a_{ij}^{k}e_{ij}\otimes
	S^{k}e_{ij}S^{\dagger k},
	\ee
	where $\{e_{ij}\}_{i,j=0}^{d-1}$ is standard matrix basis in $M(d,\mathbb{C}),$ $A^{k}=(a_{ij}^{k})\in M(d,\mathbb{C})$ are arbitrary matrices and $S$ is a matrix generator of cyclic group of
	order $d$, which has the followin properties
	\be
	S=\left( 
	\begin{array}{cccccc}
	0 & 0 & 0 & \cdots & 0 & 1 \\ 
	1 & 0 & 0 & \cdots & 0 & 0\\ 
	0 & 1 & 0 & \cdots & 0  & 0\\ 
	\vdots & \vdots & \vdots & \ddots & \vdots & \vdots \\ 
	0 & 0 & 0 & \cdots & 1 & 0 
	\end{array}
	\right),
	\ee
	so is unitary (in fact orthogonal) with eigenvalues $\{\epsilon
	^{k}:k=0,1,\ldots,d-1,\quad \epsilon ^{d}=1\}$ and eigenvectors $\{u_{k}=(
	\overline{\epsilon }^{kl})\in 
	\mathbb{C}^{d}: \ k,l=0,1,\ldots,d-1\}$, so that 
	\be
	Su_{k}=\epsilon^{k}u_{k}.
	\ee
	The circulant matrices are semi-positive definite iff all matrices $A^{k}$
	are semi-positive definite. In our example we will consider a special case
	of circulat matrices when $A^{k}=A$ for $k=0,1,\ldots,d-1$. In this case the
	circulant matrix takes the form     
	\be
	T=\left( 
	\begin{array}{ccccc}
	a_{00}\mathbf{1} & a_{01}S & a_{02}S^{2} & \cdots & a_{0d-1}S^{n-1} \\ 
	a_{10}S^{\dagger} & a_{11}\mathbf{1} & a_{12}S & \cdots & a_{1d-1}S^{n-2} \\ 
	\vdots & \vdots & \vdots & \ddots & \vdots \\ 
	a_{d-2 0}S^{n-2 \dagger} & a_{d-2 1}S^{n-3+} & a_{d-2 2}S^{n-4 \dagger} & \cdots & a_{d-2 d-1}S \\ 
	a_{d-1 0}S^{n-1 \dagger} & a_{d-1 1}S^{n-2 \dagger} & a_{d-1 2}S^{n-3 \dagger} & \cdots & a_{d-1 d-1}
	\mathbf{1}
	\end{array}
	\right) =(a_{ij}S^{i-j})
	\ee
	and it is clear that the block entries are normal and commute, so the matrix 
	$T$ satisfise the assuptions of Theorem~\ref{T1}. In this case the eigen-value
	matrices $M(\beta_{k}^{ij})\equiv M(\epsilon_{k})$ are the following  
	\be
	M(\epsilon^{k})=\left( 
	\begin{array}{ccccc}
	a_{00}\mathbf{1} & a_{01}\epsilon^{k} & a_{02}\epsilon^{2k} & \cdots & 
	a_{0d-1}\epsilon^{(n-1)k} \\ 
	a_{10}\overline{\epsilon }^{k} & a_{11}\mathbf{1} & a_{12}\epsilon^{k} & \cdots & 
	a_{1d-1}\epsilon^{(n-2)k} \\ 
	\vdots & \vdots & \vdots & \ddots & \vdots \\ 
	a_{d-20}\overline{\epsilon }^{k(n-2)} & a_{d-21}\overline{\epsilon }^{k(n-3)}
	& a_{d-22}\overline{\epsilon }^{k(n-4)} & \cdots & a_{d-2d-1}\epsilon ^{k} \\ 
	a_{d-10}\overline{\epsilon }^{k(n-1)} & a_{d-101}\overline{\epsilon }
	^{k(n-2)} & a_{d-12}\overline{\epsilon }^{k(n-3)} & \cdots & a_{d-1d-1}\mathbf{1}
	\end{array}
	\right) =(a_{ij}\epsilon ^{k(i-j)})
	\ee
	and one can check that they satisfy a nice relation 
	\be
	M(\epsilon ^{k})=A\circ u_{k}u_{k}^{\dagger},
	\ee
	where $\circ $ means the Hadamard product. According to Theorem~\ref{T1} the
	circulant matrix $T$ is separable iff it is semi-positive definite,
	equivalently when the matrices $M(\epsilon ^{k})$ are positive semidefinite,
	which holds when the matrix $A$ is semi-positive definite. In this case we
	have   the following separability decomposition of length $d$ 
	\be
	T=\sum_{k=0}^{d-1}(A\circ u_{k}u_{k}^{\dagger})\otimes u_{k}u_{k}^{\dagger}.
	\ee
\end{example}

\begin{remark}
	\label{R3}
	If in the last example the matrix $A=(a_{ij})$ is such that $a_{ij}=a$ for $%
	i,j=0,\ldots,d-1$, then the circulant matrix $T$ is a block Toeplitz matrix and
	it is known~\cite{Gur} that Toeplitz block matrices are separable only if
	they are semi-positive definite and for these matrices this holds without
	any assumptions on the block entries of the matrices.
\end{remark}

\section{Conclusions}
Summarizing in this paper we investigate separable decomosition of certain class of block matrices $T$ given by formula~\eqref{blockM}, i.e. when all sublocks are normal in $T$ and commute. Authors for such class of block matrices  prove the  Theorem~\ref{T1} containing main result for this manuscript. Mentioned theorem presents \textit{if and only if} conditions for separability of mentioned class and gives explicit method for construction of separable decomposition. What is the most important here, such decomposition can be written in terms of eigenvalues and eigenvectors in the first part and eigenvectors of the block entries in the second one as in equation~\eqref{ds}.  

%We emphasize here once again Normally is really hard to present separability conditions and explicit separable decomposition even for some restricted class of states (see Introduction), so this result is valuable and in authors opinion can have at least some small insights into separability problem.

\section*{Acknowledgments}
A. Rutkowski was supported
by a postdoc internship decision number DEC\textendash{}
2012/04/S/ST2/00002, from the Polish National Science Center. M. Mozrzymas was supported by National Science Centre project Maestro DEC-2011/02/A/ST2/00305. M. Mozrzymas would like to thank National Quantum Information Centre in Gda\'{n}sk for hospitality where some part of this work was done.

\section{Appendix: Proofs of the theorems from the main text}
\label{appA}

In this additional section we present proofs from the main text. First we give proof of the Proposition~\ref{P1} and then we present argumentation for Lemma~\ref{L1}.
\begin{proof}[Proof of the Proposition~\ref{P1}]
	The matrix $T$ has the following tensor structre 
	\be
	T=\sum_{i,j=1}^{n}E_{ij}\otimes B_{ij}, 
	\ee
	where $\{E_{ij}\}_{i,j=1}^{n}$ is a standard matrix basis and from the
	assumptions concerning the commutativity of the matrices $B_{ij}$ for $i=1,\ldots,n$
	we have 
	\be
	\forall i,j=1,\ldots,n\quad U^{\dagger}B_{ij}U=\operatorname{Diag}(\beta _{1}^{ij},\beta
	_{2}^{ij},...,\beta _{d}^{ij})\equiv D_{ij},
	\ee
	where 
	\be
	U=[u_{1}u_{2}...u_{d}] 
	\ee
	is an unitary matrix whose coloumns are common orthonormal eigenvectors of
	the matrices $B_{ij}.$ From this we have 
	\be
	\widetilde{T}\equiv (\mathbf{1\otimes }U^{\dagger})T(\mathbf{1\otimes }U)=\left( 
	\begin{array}{ccccc}
		D_{11} & D_{12} & D_{13} & \cdots & D_{1n} \\ 
		D_{21} & D_{22} & D_{23} & \cdots & D_{2n} \\ 
		\vdots & \vdots & \vdots & \ddots & \vdots \\ 
		D_{n-11} & D_{n-12} & D_{n-13} & \cdots & D_{n-1n} \\ 
		D_{n1} & D_{n2} & D_{n3} & \cdots & D_{nn}
	\end{array}
	\right), 
	\ee
	where all matrices $D_{ij}$ for $i,j=1,...,n$ are diagonal with eigenvalues of $
	B_{ij}$ on the diagonal. From the structure of the matrix $\widetilde{T}$ it
	is clear that we have the following decomposition of this matrix 
	\be
	\widetilde{T}\equiv \sum_{k=1}^{d}M(\beta _{k}^{ij})\otimes E_{kk}. 
	\ee
	Now using the identity 
	\be
	UE_{kk}U^{\dagger}=u_{k}u_{k}^{\dagger} 
	\ee
	in the equation 
	\be
	T=(\mathbf{1\otimes }U)\widetilde{T}(\mathbf{1\otimes }U^{\dagger})=
	\sum_{k=1}^{d}M(\beta _{k}^{ij})\otimes UE_{kk}U^{\dagger} 
	\ee
	we get the decomposition of the matrix $T$ given in the Proposition.
\end{proof}

\begin{proof}[Proof of the Lemma~\ref{L1}]
		It is clear, that if the matrices $M_{k}$ are semi-positive definite then
		the matrix $T$, as the sum of semi-positive matrices, is semi-positive
		definite.
		
		Suppose now that the matrix $T$ is semi-positive definite and let $
		\{v_{i}\}_{i=0}^{n-1}$ be an orthonormal basis in the space $\mathbb{C}^{n}$, then 
		\be
		\forall X=\sum_{j=1}^{d}\sum_{k=0}^{n-1}x_{kj}v_{k}\otimes u_{j}\in \mathbb{C}^{nd}\qquad (X,TX)\equiv X^{\dagger}TX\geq 0, 
		\ee
		where $x=(x_{ij})\in M(n\times d,\mathbb{C})$ is an arbitrary matrix. Now we have 
		\be
		\begin{split}
			X^{\dagger}TX&=\left( \sum_{s=0}^{n-1}\sum_{i=1}^{d}\overline{x}_{si}v_{s}^{\dagger}\otimes
			u_{i}^{\dagger}\right) \left( \sum_{k,j=1}^{d}\sum_{t=0}^{n-1}x_{tj}M_{k}v_{t}\otimes
			u_{k}u_{k}^{\dagger}(u_{j})\right) =\\	
			&=\sum_{k=1}^{d}\sum_{s,t=0}^{n-1}\overline{x}_{sk}x_{tk}(v_{s},M_{k}v_{t})=%
			\sum_{k=1}^{d}(y_{k},M_{k}y_{k}), 
		\end{split}
		\ee
		where $y_{k}=\sum_{p_{=0}}^{n-1}x_{pk}v_{p}\in\mathbb{C}^{n}$ and when the matrices $x=(x_{ps})$ run over $M(n\times d,\mathbb{C}),$ then they generates all possible sets of $\ d$ vectors $
		\{y_{k}\}_{k=1}^{d},$ $y_{k}\in \mathbb{C}^{n}.$ Now we choose a particular vector $X\in \mathbb{C}^{nd},$ defined by a particular matrix 
		\be
		x=(x_{sp})=\left\{\begin{array}{l} x_{sp}=0\quad for\quad p\neq k,\\
			x_{sk}=x_{s}\in 
			\mathbb{C}, \end{array}\right.
		\ee
		where $k\in \{1,\ldots,d\}$ and is arbitrary. So in the matrix $x$ only the $k^{\text{th}}$ coloumn is not equal to zero and it forms an arbitrary vector from the
		space $\mathbb{C}^{n}$. It is clear that this particular vector $X$ defines the particular
		set of vectors $\{y_{1}=0,y_{2}=0,\ldots,y_{k},\ldots,y_{d}=0\}$, where $y_{k}\in \mathbb{C}^{n}$ is arbitrary.  From semi-positivity of the matrix $T$ for this
		particular vector $X$ we have 
		\be
		0\leq X^{\dagger}TX=\sum_{s=1}^{d}(y_{s},M_{s}y_{s})=(y_{k},M_{k}y_{k}), 
		\ee
		where $k\in \{1,\ldots,d\}$ and is arbitrary and $y_{k}\in \mathbb{C}^{n}$ is also arbitrary so it means that all matrices $M_{k}$ are
		semi-positive.
\end{proof}


\begin{thebibliography}{9}
		
\bibitem{Ein} A. Einstein, B. Podolsky and N. Rosen, \textit{Can Quantum-Mechanical Description of Physical Reality Be Considered Complete?}, Phys. Rev. A, vol. 47, 777, 1935

\bibitem{Bennett2} C.~H. Bennett, G. Brassard. Cr\'epeau, R. Jozsa, A. Peres, W.~K. Wootters,  \textit{Teleporting an Unknown Quantum State via Dual Classical and Einstein-Podolsky-Rosen Channels}, Phys. Rev. Lett. {\bf 70} 1895-1899 (1993).

\bibitem{Bennett3} C.~H. Bennett, S. Wiesner \textit{Communication via one- and two-particle operators on Einstein-Podolsky-Rosen states}, Phys. Rev. Lett. {\bf 69} (20), 2881, (1992).

\bibitem{Ekert1} A.Ekert,  \textit{Quantum cryptography based on Bell’s theorem}, Phys. Rev. Lett. {\bf 67} 661-663 (1991).

\bibitem{ryszh} R. Horodecki, P. Horodecki, M. Horodecki and K. Horodecki, \textit{Quantum Entanglement}, Rev. Mod. Phys. {\bf 81}: 865-942,2009 


\bibitem{Gur} L. Gurvits and H. Barnum, \textit{Largest separable balls
	around the maximally mixed bipartite quantum state}, Phys. Rev. A, vol. 66,
p. 062311, 2011

\bibitem{Han} P. Ch. Hansen, \textit{Deconvolution and regularization with Toeplitz matrices}, Numerical Algorithms, vol. 29, p. 323-378, 2002

\bibitem{Xiao} N. Johnston, \textit{Covariance matrix estimation for stationary
time series}, The Annals of Statistics, vol. 40, No.1, p. 466-493, 2012

\bibitem{Reddi} S. S Reddi, \textit{Eigenvector properties of Toeplitz matrices and their application to spectral analysis of time series}, Signal Processing, vol. 7, p. 45-56, 1984

\bibitem{Chrus} D. Chru{\'s}ci{\'n}ski and A. Kossakowski, Phys. Rev. A {\bf 76}, 032308 (2007).

\bibitem{Cho92} S. J. Cho, S-H. Kye i S. G. Lee. \textit{Generalized choi maps in threedimensional
matrix algebra}, Linear Albebra Appl. {\bf 171}, 213, (1992).

\bibitem{PH1} P. Horodecki, \textit{Separability criterion and inseparable mixed states with positive partial transposition}Phys. Rev. Lett. A {\bf 232}, 333, (1997).

\bibitem{PH2} R. Horodecki, P. Horodecki, M. Horodecki, K. Horodecki, \textit{Quantum entanglement}, Rev. Mod. Phys. {\bf 81}, 865, (2009).

\bibitem{Horn} R. A. Horn, Ch. R. Johnson, \textit{Topics in Matrix Analysis}, Cambridge University Press, 1991

\bibitem{Kie} K.-Ch. Ha, S.-H. Kye, \textit{Separable states with unique decompositions}, Communications in Mathematical Physics, Volume 328, Issue 1, pp 131-153 (2014) 

\bibitem{Chrus2} D. Chru{\'s}ci{\'n}ski, J. Jurkowski and A. Rutkowski, \textit{A class of bound entangled states of two qutrits},  Open Sys. Information Dyn. 16, 235-242, (2009).

\bibitem{Wer} R. F. Werner, \textit{Quantum states with einstein-podolsky-rosen correlations admitting a hidden-variable mode}, Phys. Rev. A {\bf 40}, 4277,
(1989).

\bibitem{Hor2} M. Horodecki and P. Horodecki, \textit{Reduction criterion of separability and limits for a class of distillation protocols}, Phys. Rev. A {\bf 59}, 4206, (1999).

\bibitem{Chen} L. Chen, E. Chitambar, K. Modi, and G. Vacanti, \textit{Multipartite Classical States and Detecting Quantum Discord}, Phys. Rev. A {\bf 83}, 020101 (2011).

\bibitem{Zurek}   H. Ollivier and W. H. Zurek, \textit{Quantum Discord: A Measure of the Quantumness of Correlations}, Physics Review Letters vol. 88, 017901 (2001) 

\bibitem{Baum} B. Baumgartner, B.C. Hismayr and H. Narnhofer, \textit{The state space for two qutrits has a phase space structure in its core}, Phys. Rev. A {\bf 74}, 032327, (2006)

\bibitem{Baum2} B. Baumgartner, B.C. Hismayr and H. Narnhofer, \textit{A special simplex in the state space for entangled qudits}, J. Phys. A: Math. Theor. {\bf 40}, 7919, (2007)

\bibitem{Baum3} B. Baumgartner, B.C. Hismayr and H. Narnhofer, \textit{The geometry of bipartite qutrits including bound entanglement}, Phys. Lett. A {\bf 372}, 2190, (2008)

%\bibitem{Bini1} D.A. Bini, B. Meini, \textit{On the solution of a nonlinear matrix equation arising in queueing problems}, SIAM J. Matrix Anal. Appl., vol. 17, p. 906–926, 1996

\end{thebibliography}
\end{document}